\documentclass[11pt,onecolumn]{article}
\usepackage{amsmath,amsthm,amssymb,amsfonts,epsfig,graphicx}
\usepackage{algorithm,algorithmicx}
\usepackage{color}
\usepackage{dsfont}
\usepackage{wrapfig}

\usepackage[left=1.0in,top=1.0in,right=1.05in,bottom=0.5in,nohead,textheight=10in,footskip=0.3in]{geometry}

\newtheorem{theorem}{Theorem}

\newtheorem{lemma}[theorem]{Lemma}
\newtheorem{definition}[theorem]{Definition}

\begin{document}

\long\def\ignore#1{}
\def\myparagraph#1{\vspace{2pt}\noindent{\bf #1~~}}



\newcommand{\cs}{{coalitional structure}}
\newcommand{\CS}{\mathit{CS}}
\newcommand{\calCS}{{\cal CS}}
\newcommand{\tup}[1]{\left<{#1}\right>}
\newcommand{\np}{$\mathcal{NP}$}
\newcommand{\p}{$\mathcal{P}$}



\title{Optimal Coalition Structures in Cooperative Graph Games\footnote{A short version of this paper is to appear at AAAI 2013}}
\author{Yoram Bachrach$^{\dag}$, Pushmeet Kohli$^{\dag}$, Vladimir Kolmogorov$^{\ddag}$ and Morteza Zadimoghaddam$^{\S}$ \\
\\
\normalsize
$\dag$ Microsoft Research, Cambridge, UK \\
$\ddag$ Institute of Science and Technology (IST), Austria \\
$\S$ MIT, Cambridge, MA, USA}

\date{}
\maketitle

\begin{abstract}
Representation  languages for coalitional games are a key research area in algorithmic game theory. There is an inherent tradeoff between how general a language is, allowing it to capture more elaborate games, and how hard it is computationally to optimize and solve such games. One prominent such language is the simple yet expressive \emph{Weighted Graph Games (WGGs)} representation~\cite{deng_papa_1994}, which maintains knowledge about synergies between agents in the form of an edge weighted graph.

We consider the problem of finding the optimal coalition structure in WGGs. The agents in such games are vertices in a graph, and the value of a coalition is the sum of the weights of the edges present between coalition members. The \emph{optimal coalition structure} is a partition of the agents to coalitions, that maximizes the sum of utilities obtained by the coalitions. We show that finding the optimal coalition structure is not only hard for general graphs, but is also intractable for restricted families such as planar graphs which are amenable for many other combinatorial problems. We then provide algorithms with constant factor approximations for planar, minor-free and bounded degree graphs.
\end{abstract}

\section{Introduction}
\label{l_introduction}

Consider a set of agents who can work in teams. Some agents work well together, while others find it hard to do so. When two agents work well together, a team which contains both of them can achieve better results due to the synergy between them. However, when agents find it hard to work together, a team that contains both agents has a reduced utility due to their inability to cooperate, and may perform better when one of them is removed. How should we best partition agents into teams to maximize the total utility generated?


Cooperation is a central issue in 
algorithmic game theory, 
and \emph{cooperative games} are very useful for modeling team formation and negotiation in many domains. 
In such games, agents form coalitions to pursue a joint cause, and must decide how to split the gains they derive as a group. Much of previous literature explores settings where only one coalition can be formed. However, in many scenarios agents can form \emph{multiple} disjoint coalitions, where each coalition can obtain its profits independently. An important goal is to partition the agents in a way that maximizes the total value gained by all coalitions, i.e. the social welfare. This problem is known as \emph{finding the optimal coalitional structure}.

Cooperative game theory provides tools to reason about how to best partition the agents into team or how the utility generated by an agent team should be allocated to its members. The key piece of information used by such game theortic tools is the amount of utility any team of agents could potentially generate. Since an agent team, sometimes called a \emph{coalition} is simply a subset of agents, the number of different coalitions is exponential in the number of agents. The mapping between any agent coalition and the utility it can generate lies at the heart of any cooperative game, and is called the \emph{characteristic function} of the game. 

One possible way to represent the characteristic function is by simply listing down this utility for any possible coalition. Generally any coalition may have a different value, so a na\"ive representation requires storage exponential in the number of agents. This emphasizes the need for a \emph{succinct representation}. 

In many domains it is possible to use \emph{knowledge} about \emph{specific features} of the domains and provide a more succint representation of the characteristic function. However, even for such succint representations, reasoning about the game may still be computationally hard. Previous work in algorithmic game theory has examined many such representation languages for cooperative games and the computational complexity of calculating various solutions in them. 
For general surveys of such representations, see~\cite{coopgt-book,shoham2009multiagent,bilbao:2000:cooperative}. Some approaches guarantee a polynomial description, but only represent restricted games
~\cite{peleg2007introduction,deng_papa_1994}. Others can represent any game, but require exponential storage in the worst case~\cite{IeoSho06,bachrach2010coalitional}. We examine coalition structures in the prominent \emph{Weighted Graph Games (WGG)} model of~\cite{deng_papa_1994}. In WGGs, agent synergies are expressed using a graph, where agents are vertices and the weight of an edge connecting two agents expresses how well they work together. A positive weight indicates they can coordinate well, yielding a positive contribution to a coalition containing both. The edge's weight expresses how much utility can be derived from the cooperation of the two agents. A negative weight indicates the agents do not work well together, diminishing the utility of a team containing both. Again, the edge's weight expresses the reduction in utility of a coalition that contains both agents. This graph representation allows expressing synergies in coalitions that agents form.\footnote{Group buying sites (e.g. LivingSocial and Groupon) reward social recommendations, so WGGs can capture synergies from including enough friends of a consumer to make her buy a good.}

\emph{Our contribution}: We study optimal coalition structures in WGGs. We prove that finding the optimal coalition structure in WGGs is hard even for restricted families of graphs such as planar graphs. We provide constant factor approximation algorithms for planar, minor-free and bounded degree graphs.

Note that the objective function that we study coincides with the {\em Correlation Clustering} functional~\cite{bansal2004cc} {\bf up to an additive constant}.
Due to this additive shift existing approximability results for Correlation Clustering such as those in~\cite{tan2008,mitra2009}
do not translate to our problem. We believe that the additive normalization that we use is quite natural in our context.
To our knowledge, we are the first to study approximation schemes for this functional. 

{\bf Preliminaries:}
A transferable utility (TU) coalitional game is composed of a set of $n$ agents, $I$, and a characteristic function mapping any subset (coalition) of the agents to a rational value $v : 2^I \rightarrow \mathbb{Q}$, indicating the total utility these agents achieve together. We follow the \emph{Coalition Structure (CS)} model of~\cite{aumann1974cooperative}, where agents can form several teams \emph{simultaneously}. A coalition structure is a partition of the agents into disjoint coalitions. Formally, $\CS=(C^1, \dots, C^l)$ is a \emph{coalition structure} over $I$ if $\cup_{i=1}^l C^i=I$ and $C^i\cap C^j=\emptyset$ for all $i\neq j$; The set $\calCS(I)$ denotes all possible coalition structures on $I$.

We overload notation and denote $v(\CS) = \sum_{C^j\in \CS}v(C^j)$. Our focus in this paper is on the coalitional structure generation problem, of finding the \emph{optimal coalition structure} $CS^*$, with maximal value. 
Given a game $\tup{I,v}$, an \emph{optimal} coalition structure $CS^* \in \calCS(I)$ is a partition that maximizes welfare, i.e. for any other $CS \in \calCS(I)$ we have $v(CS) \leq v(CS^*)$. We denote the problem of finding an optimal coalition structure as \emph{OPT-CS}. OPT-CS is equivalent to a complete set partitioning problem where all disjoint subsets of the set of all agents are possible. The set partitioning problem was studied in~\cite{YunYeh1986} and shown to be NP-hard. However, we focus on OPT-CS where inputs are restricted to WGGs~\cite{deng_papa_1994}.

We now review the \emph{Weighted Graph Games (WGGs)} model~\cite{deng_papa_1994}.
\begin{definition}
\label{l_def_graph_games}
WGGs are games played over a graph $G=\tup{V,E}$, with edge weights $w : E \rightarrow \mathbb{Q}$. The agents are the vertices, so $I=V$, and the characteristic function is the sum of the weights on the graph induced by the coalition. Given a coalition $C \subseteq V$, we denote the edges induced by the coalition as $E_C = \{ e = (u,v) \in E | u,v \in C \}$. The characteristic function is 
$v(C) = \sum_{e \in E_C} w(e)$.
\end{definition}

\noindent As noted in~\cite{deng_papa_1994} WGG cannot represent all games. We allow at most one edge between any two vertices (parallel edges may be merged, summing the weights).

\section{Finding the Optimal Coalitional Structure}
\label{l_main_res}
We first formally define the OPT-CS problem. 

\begin{definition} [OPT-CS-WGG]
  \label{l_opt_cs}
Given a WGG $\tup{I,v}$ and a rational number $r$, test if the value of the optimal coalitional structure for this game is at least $r$, i.e. if there is $CS\in\mathcal{CS}(I)$ such that $v(CS)\geq r$. We denote an optimal structure as $CS^* \in \arg \max_{CS \in \calCS(I)} v(CS)$.
\end{definition}

\subsection{Hardness Results}
\label{l_hardness_results}

We first discuss why OPT-CS-WGG is hard. As our contribution lies in providing tractable algorithms for restricted cases, we only provide sketches for hardness results.
%
It is quite easy to show hardness for general graphs using a reduction from Independent Set 
(IS). IS is the problem of testing if there is a subset of vertices of size $k$ with no edges between them in an input graph $G=\tup{V,E}$. Theorem~\ref{theorem_planar_hardness} is a stronger result, so we only provide a sketch of the proof (for completeness a detailed proof is given in the appendix).

\begin{theorem}OPT-CS-WGG is \np-complete, and assuming $\mathrm{P} \neq \mathrm{NP}$, there is no polynomial time $O(n^{1/2-\epsilon})$-approximation algorithm for it where $n$ is the number of vertices and $\epsilon$ is any positive constant. There is no polynomial time $O(n^{1-\epsilon})$-approximation algorithm for this problem unless $\mathrm{NP}=\mathrm{ZPP}$.
\label{l_thm_opt_cs_hard_stsg}
\end{theorem}
\begin{proof}

We reduce an IS instance $\tup{G=\tup{V,E},k}$ to an OPT-CS-WGG instance. The reduced graph $G'=\tup{V',E'}$ has all vertices and edges of $G$ (so $V \subseteq V'$ and $E \subseteq E'$), and an additional vertex $s$. The weights of the original edges are all set to be $-k$ (where $k > |V|$). Vertex $s$ is connected to each vertex in $v \in V$ with an edge $e$ of weight $w(e) = 1$. $G$ has an independent set of size $k$ iff there exists a partition for $G'$ with value $k$. This is a parsimonious reduction between IS and OPT-CS-WGG. H{\aa}stad proved that there is no polynomial time $O(n^{1/2-\epsilon})$-approximation algorithm for IS assuming $\mathrm{P} \neq \mathrm{NP}$, and no polynomial time $O(n^{1-\epsilon})$-approximation algorithm for IS unless $\mathrm{NP}=\mathrm{ZPP}$ \cite{Hastad96}.
\end{proof}

The above result shows that OPT-CS-WGG has no good approximation for general graphs. It might seem that the hardness comes from having to avoid all negative edges, as we can make the weights of the negative edges very low to make sure that they are not present in the optimal solution. However, we show OPT-CS-WGG is hard and inapproximable even when all weights have the same absolute value, using an involved reduction from IS to the problem of OPT-CS-WGG where all edges are either $+1$ or $-1$, denoted OPT-CS-WGG$\pm 1$. 
\begin{theorem}OPT-CS-WGG$\pm 1$ is \np-complete. Assuming $\mathrm{P} \neq \mathrm{NP}$, there is no polynomial time $O(n^{1/2-\epsilon})$-approximation algorithm for it where $n$ is the number of vertices of the graph, and $\epsilon$ is any positive constant. There is no polynomial time $O(n^{1-\epsilon})$-approximation algorithm for this problem unless $\mathrm{NP}=\mathrm{ZPP}$. 
\label{l_thm_opt_cs_hard_stsg_plus_minus_one}
\end{theorem}
\begin{proof}
Similarly to the proof of Theorem~\ref{l_thm_opt_cs_hard_stsg}, we reduce an Independent Set instance $G=\tup{V,E}$ to a OPT-CS-WGG instance. The reduced graph $G'=\tup{V',E'}$ contains all the vertices and edges of $G$ (so $V \subseteq V'$ and $E \subseteq E'$), and one additional vertex $s$. The weights of the original edges are identical, and are all $w(e) = -1$ (this is where the reduction differs from that of Theorem~\ref{l_thm_opt_cs_hard_stsg}). Vertex $s$ is connected to any vertex in $v \in V$ with an edge $e$ with weight $w(e) = +1$.

If the optimal coalition structure has value $k'$ in graph $G'$, we show we can find an independent set of size at least $k'/9$ in $G$.
As in the previous reduction, all positive edges are incident to vertex $s$. Thus every vertex is either in the coalition of vertex $s$ or in a separate single-vertex coalition. Let $S$ be the set of vertices from graph $G$ in the coalition of vertex $s$. The value of coalition $S$ is thus the value of the entire coalition structure. Suppose there are $a$ vertices in $S$, and $b$ negative edges between vertices of $S$. In this case, the sum of all weights of positive edges for $S$ is equal to $a$, and the sum of all weights of negative edges in $S$ is $-b$. Therefore, $k'=a-b$, so $a > k'$. Also note that $b < a$.

Now consider the subgraph $G[S]$. The number of edges in this graph, $b$, is not more than the number of vertices $a$. Therefore, the average degree is at most $2$ in this subgraph. Thus, the number of vertices with degree at least $3$ in this subgraph is not more than $2a/3$, so there are at least $a/3$ vertices with degree at most $2$. 

Let $S'$ be the set of vertices with degree at most $2$ in the subgraph $G[S]$. Now consider the subgraph $G[S']$. This subgraph has at least $a/3$ vertices, each with a degree at most $2$. We can pick $1/3$ of the vertices of this subgraph as an independent set: we just pick a vertex arbitrarily, put it in our independent set and remove its neighbors, which are no more than two vertices. Thus, in the worst case, for every three vertices, we choose one for our independent set. This way, we can find an independent set of size at least $a/9 \geq k'/9$, so our reduction loses a factor of at most $9$, but the same hardness results hold asymptotically.
\end{proof}


OPT-CS-WGG remains hard even for planar graphs. Independently of us ~\cite{DBLP:journals/corr/abs-1102-1747} examine a related coalition structure generation problem for a graph representation language. Their representation is combinatorially richer, so hardness results do not generally carry over to the simpler WGG representation. However, their proof that generating the optimal coalition structure where their graph representation is planar does carry over to the simplified WGG setting. We also provide an alternative proof in the appendix. This yields Theorem~\ref{theorem_planar_hardness}. 

\begin{theorem}
OPT-CS-WGG in planar graphs is \np-complete.
\label{theorem_planar_hardness}
\end{theorem}
\begin{proof}
Given in ~\cite{DBLP:journals/corr/abs-1102-1747}; An alternative proof is given in the appendix.  
\end{proof}

\subsection{Exact Algorithms for Bounded Treewidth Graphs}
\label{l_positive_results}
For completeness, we first consider OPT-CS-WGG restricted to graphs of bounded treewidth.
Lemma~\ref{l_thm_opt_cs_tree_p} is a special case of Lemma~\ref{l_thm_opt_cs_bounded_treewidth_p} from~\cite{Cowans05}.\footnote{We decided to include Lemma~\ref{l_thm_opt_cs_tree_p} since (i) its proof is much simpler than that of
Lemma~\ref{l_thm_opt_cs_bounded_treewidth_p}; 
(ii) it provides a distributed (local) algorithm for partitioning. Two neighbors are in the same coalition if and only if their connecting edge has positive weight;
(iii) it shows that for forests, there is a simple solution that achieves all positive edges, and avoids all negative edges.}

\begin{lemma}OPT-CS-WGG with inputs restricted to trees (or forests) is in \p. 
\label{l_thm_opt_cs_tree_p}
\end{lemma}
\begin{proof}
Let $P \subseteq E$ be the set of edges with positive weights, and $N \subseteq E$ be the edges with negative weights. For an edge subset $A$ we denote its total weight as $w(A) = \sum_{e \in A} w(e)$. Note that for any structure $CS$ we have $v(CS) \leq w(P)$. We show that for trees the optimal structure $CS^*$ has $v(CS^*) = w(P)$. A very simple partitioning to two sub coalitions $X,Y$ suffices for this. We begin with an arbitrary leaf $v$, and put it in one of the sub coalitions, so $v \in A$. For any neighbor $u$ of $v$, so $e=(v,u) \in E$, we choose a sub coalition for $u$ based on $w(e)$. If $w(e) > 0$ we put $u$ in the same sub coalition as $v$, and if $w(e) < 0$ we put $u$ in the other sub coalition. We then continue the process from the vertex $u$, until we deplete all the vertices in the graph. Since the graph is a tree, any edge with negative weight has one vertex in $X$ and the other in $Y$, and is not counted for the value of the coalition structure, while every positive weight edge has both its vertices in the same sub coalition, and is counted for the value of the coalition structure. Thus we have $v(CS^*) = w(P)$. The partitioning can be done using a simple breadth first search, in polynomial time of $O(|V|+|E|)$.
\end{proof}

\begin{lemma}[\cite{Cowans05}]For $k$-bounded treewidth graphs (constant $k$), OPT-CS-WGG is in \p.
\label{l_thm_opt_cs_bounded_treewidth_p}
\end{lemma}


\subsection{Constant Factor Approximation for Planar Graphs}
We provide a polynomial time constant approximation algorithm for planar graph games. Planar graphs model many real-world domains. For example, consider coalitions between countries. In many domains, synergies would only exist between neighbouring countries. We can define a graph game where countries have an edge between them only if they are neighbours, resulting in a \emph{planar} graph. Our method achieves a coalition structure avoiding all negative edges and gains a constant portion of the weights of the positive edges. The optimal value is at most the sum of the positive weights, yielding a constant approximation. First we define feasible sets which play an important role in our algorithm.

\begin{definition}
Given a planar graph $G$, denote by $E^+$ the set of positive edges, and by $E^-$ the set of negative edges. A subset $E' \subseteq E^+$ is \emph{a feasible set} iff there is a partition $P$ of vertices such that any edge $e \in E'$ is contained in one part of the partition, while all negative edges are cross-part edges. We say $P$ achieves $E'$. 
\end{definition}

We find partitions that achieve feasible sets $E_1, E_2, \cdots, E_k$ (each a subset of $E^+$), whose union $\cup_{i=1}^k E_i$ is $E^+$. Each positive edge is thus achieved by at least one of these $k$ partitions. The value of partition $i$ is at least $\sum_{e \in E_i} w(e)$ as we avoid all negative edges in our partitions, making the value of a partition be the sum of the positive weights achieved by it. The union of these $k$ feasible sets is the set of all  positive edges. Thus the sum of the values of the $k$ partitions is at least the sum of all positive weights. Picking the maximal value partition, we obtain a $k$-approximation algorithm. Our algorithm finds \emph{few} such partitions that still \emph{cover all positive edges}. 
In our algorithm we present a constant number of these partitions (feasible sets). We start with building some principal feasible sets that we use later in our algorithm. The first building block (feasible set) is a matching.
\begin{lemma}\label{lemma:MatchingOnePartition}
Every matching $M \subseteq E^+$ is a feasible set. In other words, there is a partition that avoids all negative edges, and achieve edges of $M$.
\end{lemma}
\begin{proof}
Let $e_1, e_2, \cdots, e_a$ be the edges of a matching. We build a partition of the vertices. We have $a$ clusters for the $a$ edges of our matching. We put both endpoints of $e_i$ in cluster $i$. There are $n-2a$ remaining vertices as well. We put them in $n-2a$ separate single-vertex clusters, so we achieve the edges of $M$ in this partition. We show we avoid all negative edges. Suppose not, so there is a negative edge $e'(u,v)$ such that $u$ and $v$ are in the same cluster. Thus $u$ and $v$ are endpoints of an edge in the matching as well, so there are two edges between $u$ and $v$ in to contradiction to our assumption of having no parallel edges.
\end{proof}

The second building block is slightly more elaborate. We prove that the union of vertex-disjoint stars can be covered using at most three feasible sets. A star is a subgraph formed of several edges that all share one endpoint called center vertex. In other words, a star is a vertex with some edges to some other vertices, and there are no other edges in the star between vertices.
\begin{lemma}
\label{lemma:StarswithThreePartition}
We can cover a union of several vertex-disjoint stars using at most 3 feasible sets. 
\end{lemma}
\begin{proof}
Assume there are $l$ stars with center vertices $v_1, v_2, \cdots, v_l$. In the star $i$, vertex $v_i$ has edges to vertex set $S_i$. The sets $S_1, S_2, \cdots, S_l$ are disjoint and they do not contain any of the center vertices.
We know that graph $G$ is planar, and therefore we can find a proper four-coloring for it. Consider vertex $v_i$. Non of the vertices in set $S_i$ has the color of vertex $v_i$, so the vertices of set $S_i$ are colored using only three colors. Without loss of generality, assume they are colored with colors $1, 2, 3$. Let $S_{i,1}$ be the set of vertices in $S_i$ with color one. Similarly we define sets $S_{i,2}$ and $S_{i,3}$. We do this for all center vertices.
Note that there is no edge inside set $S_{i,j}$ for $1 \leq i \leq l$, and $1 \leq j \leq 3$. But there might be edges between different sets.
Now here is the first partition that gives us the first feasible set. We put $v_i$ and all vertices of $S_{i,1}$ in one cluster, and we do this for all center vertices. So for each center vertex, we have a separate cluster. All remaining vertices of the graph go to separate single-vertex clusters. We achieve the edges between vertex $v_i$, and all vertices in $S_{i,1}$ for $1 \leq i \leq l$. We avoid all negative edges because there is no edge inside set $S_{i,1}$, and there is no negative edge between vertex $v_{i}$, and set $S_{i,1}$.
Similarly we use two other partitions to get the edges between $v_i$ and sets $S_{i,2}$ and $S_{i,3}$. So we can cover all these edges using three feasible sets.
\end{proof}

\begin{lemma}\label{lemma:Forest6Partitions}
The edges of a forest can be covered using 6 feasible sets. 
\end{lemma}
\begin{proof}
We show the proof for one tree. The same should be done for other trees.
Make the tree $T$ rooted at some arbitrary vertex $r$. Now every vertex in the tree has a unique path to $r$. We color the edges of this tree using two colors, red and blue. The edges that have an odd distance to $r$ are colored red, and the edges with even distance to $r$ are blue. So the first level of edges that are adjacent to $r$ are colored blue, the next level edges are red, and so on. So we start from $r$, and alternatively color edges blue and red.
Considering the subgraph of blue edges, we cannot find a path of length four (3 edges) in it, and we know that there is no cycle in the graph. Thus, these blue edges can form only some disjoint stars. The same proof holds for the red edges. 
Using Lemma~\ref{lemma:StarswithThreePartition}, we cover the blue edges with three feasible sets, and the red edges with another three feasible sets. So, using at most $6$ feasible sets, every edge is covered in the tree.
\end{proof}


Now we just need to show that the set of positive edges in $G$ can be decomposed into a few number of forests. This can be implied by a direct application of Nash-Williams Theorem \cite{Nash-Williams}.
Let $G^+$ be the subgraph of $G$ with all positive edges. Clearly $G^+$ is also planar. Nash-Williams Theorem states that the minimum number of forests needed to partition the edges of a graph $H$ is equal to $\max_{S \subseteq V(H)} \frac{m_S}{n_S-1}$ where $m_S$, and  $n_S$ are the number of edges and vertices in set $S$ respectively. $V(H)$ is the set of all vertices in graph $H$. Since $G^+$ is planar, $m_S$ is at most $3n_S-6$ for every $S \subseteq V(G^+)$ which implies that three forests are sufficient to cover all edges of $G^+$. We should note that computing these forests can be done in polynomial time as Gabow and Westerman \cite{GabowWesterman} provided a polynomial time algorithm that makes Nash-Williams theorem constructive. We conclude that all positive edges in $G$ can be covered with $3 \times 6 = 18$ feasible sets which yields an $18$-approximation algorithm.

\subsection{Minor Free Graphs}
We generalize our results to Minor Free Graphs. A graph $H$ is a minor of $G$ if $H$ is isomorphic to a graph that can be obtained by zero or more edge contractions on a subgraph of $G$.
$G$ is $H$-Minor Free, if it does not contain $H$ as a minor. A graph is planar iff it has no $K_5$ or $K_{3,3}$ as a minor~\cite{Wagner37} ($K_5$ is a complete graph with $5$ vertices, and $K_{3,3}$ is a complete bipartite graph with $3$ vertices in each part). 
We give an $O(h^2\log{h})$-approximation algorithm for OPT-CS-WGG in $H$-minor free graphs where $h$ is the number of vertices of graph $H$. This yields a constant factor approximation for planar graphs in particular, because they are $K_5$-minor free graphs. We use the following theorem~\cite{Thomason01}.
The main property of Minor Free graphs that make them tractable in this problem is sparsity. 

\begin{theorem}\label{Thm:MinorFreeSparsity}
The number of edges of a $H$-Minor Free graph $G$ with $n$ vertices is not more than $cn$ where $c$ is equal to $(\alpha+o(1))h\sqrt{\log{h}}$, $h$ is the number of vertices in $H$, and $\alpha$ is a constant around $0.319$.
\end{theorem}

Our algorithm for planar graphs used \emph{sparsity} to make sure that there are low degree vertices at each level 
and we also need to make sure that the graph is colorable with a few colors. Using sparsity we can find a proper coloring of these graphs using $2c+1$ colors. Any subgraph $G'$ of $G$ is also $H$-Minor Free, so its number of edges is at most $O(c|G'|)$. Thus the average degree of every subgraph of $G$ is at most $2c$, so in every subgraph of $G$, we can find some vertices with degree at most $2c$ (the average degree). We use this to get a proper $2c+1$-coloring of $G$. Let $v$ be a vertex in $G$ with degree at most $2c$. Clearly $G \setminus \{v\}$ is also $H$-minor free, and inductively we can find a proper $2c+1$-coloring for it. Vertex $v$ has at most $2c$ neighbors, so one of the $2c+1$ colors is not used by its neighbors. Thus we can find one appropriate color for $v$ among our $2c+1$ colors, and consequently have a proper $2c+1$-coloring for $G$.

Using theorem~\ref{Thm:MinorFreeSparsity}, we know that every subset $S$ of $G$ has at most $c|S|$ edges because every subgraph of $G$ is also $H$-minor free. We can then use the Nash-Williams Theorem \cite{Nash-Williams}, and Gabow and Westerman Algorithm \cite{GabowWesterman} to cover all positive edges of $G$ with $O(c)$ forests.
%
We also know that each forest can be decomposed into two unions of stars. The entire graph $G$ is colorable using $2c+1$ colors, so we need $2c+1-1=2c$ feasible sets to cover a union of stars. 
We conclude that $2c \cdot 2 \cdot O(c) = O(c^2)$ feasible sets are enough to cover all positive edges, and get a $O(c^2)$ approximation algorithm by picking the best (maximum weight) feasible set. Since $c$ is $O(h\sqrt{\log{h}})$, our approximation factor is $O(h^2\log{h})$.

\subsection{Bounded Degree Graphs}
We now consider bounded degree graphs. A vertex's positive degree is the number of positive edges incident to it. Denote the maximum positive degree in $G$ as $\Delta$. Using the Vizing Theorem, the positive edges can be decomposed into $\Delta+1$ matchings (which are also feasible sets). This yields a $\Delta+1$ approximation algorithm. A polynomial time algorithm for finding the decomposition can be derived from the Vizing Theorem's proof. We give a \emph{linear time} algorithm for this problem: a randomized $(2+\epsilon)\Delta$ approximation with an expected running time $O(E \log{\Delta} / \epsilon)$ and $O(V+E)$ space\footnote{According to an anonymous reviewer, there are techniques to get rid of the $\log{\Delta}$ factor in the running time which results in higher space complexity of $O(E+V \cdot \Delta)$ instead of the linear space in our algorithm.} where $E$ is the number of edges. 

We pick an arbitrary ordering of the positive edges of the graph, and try to decompose them into $(2+\epsilon)\Delta$ matchings. We color the edges with $(2+\epsilon)\Delta$ colors such that no two edges with the same color share an endpoint. Assume that we are in step $i$, and wish to color the $i$-th edge in our ordering. Let $e$ be this edge with endpoints $u$ and $v$. We have already colored $i-1$ edges and some of those colored edges may have $u$ or $v$ as their endpoints, so we must avoid the color of those edges. For each vertex, we keep the color of its edges in a data structure. Initially the data structures for all vertices are empty. When we color an edge, we add its color to the data structure of its two endpoint vertices. We can use binary search trees to insert and search in $O(\log{\Delta})$ time, since we insert at most $\Delta$ colors in each data structure. For edge $e$, we must find a color that is not in the union of the data structures of vertices $u$ and $v$. There are at most $2(\Delta-1)$ colors in these two data structures, and there are $(2+\epsilon)\Delta$ colors in total. Thus if we randomly pick a color, with probability at least $\epsilon/2$, we can use this color for edge $e$. Checking whether a color is in a data structure can be done using a search query. Thus we can check in time $O(\log{\Delta})$ whether the randomly chosen color is good or not. If the color is already taken, we can try again. It takes at most $2/\epsilon$ times in expectation to find an available color. Thus for each edge we spend $O(\log{\Delta}/\epsilon)$ time to find a color, so the average running time of this decomposition is $O(E\log{\Delta}/\epsilon)$ in expectation. Finding the best feasible set (matching) does not take more than $O(E)$ time. Thus we get an $(2+\epsilon)\Delta$-approximation with almost linear running time.

\subsection{Treewidth Based Approximations}
Many hard problems are tractable for graphs with constant or bounded treewidth. We present a polynomial time $O(k^2)$-approximation algorithm where $k$ is the treewidth of the graph, without assuming that the treewidth of graph $G$ is constant or a small number. Algorithms for finding the treewidth of a graph only work in polynomial time when the treewidth is constant. Although we do not know the treewidth, we can still make sure that the approximation factor is not more than $O(k^2)$. We use the following lemma.

\begin{lemma}
\label{lemma_degree_k}
If $G$ has treewidth $k$, it has a vertex with degree at most $k$. 
\end{lemma}
\begin{proof}
Since $G$ has treewidth $k$, there is a tree $T$ such that every vertex of $T$ has a subset of size at most $k+1$ of vertices of $G$. 
The vertices of $T$ that contain a vertex of $G$ form a connected subtree. We also know that the endpoint vertices of each edge in $G$ are in the set of at least one vertex of $T$ together. Now we can prove our claim. Consider a leaf $v$ of $T$. Let $u$ be the father of $v$. These two vertices have two subsets of vertices of $G$ like $S_u$ and $S_v$. If $S_v$ is a subset of $S_u$, there is no need to keep vertex $v$ in our tree. We can delete it from $T$, and the remaining tree is also a proper representation of graph $G$. So we know that there is at least a vertex of $G$ like $x$ which is in $S_v$, and not in $S_u$. Clearly $v$ is the only vertex of $T$ that contains $x$. Otherwise the vertices of $T$ that contain $x$ do not form a connected subgraph. So vertex $x$ can have neighbors only in set $S_v$ which means that $x$ has at most $k+1-1=k$ neighbors.
\end{proof}

Removing a vertex from a graph does not increase its treewidth, so we can iteratively find vertices of degree at most $k$, and delete them. Thus we can find a $(k+1)$ proper coloring of vertices of $G$. We use the same decomposition we used for planar graphs, so we decompose the positive edges into $O(k)$ matchings and unions of stars. Each matching is a feasible set, and each union of stars can be decomposed into $k+1-1=k$ feasible sets, as we can color the graph with $k+1$ colors. Thus $O(k^2)$ feasible sets are enough to cover all positive edges yielding an $O(k^2)$ approximation algorithm. Note that we start with a value of $k$, and we keep deleting vertices with degree at most $k$. If every vertex is deleted after some number of iterations, we achieve the desired structure. Otherwise at some point the degree of each vertex is greater than $k$, indicating that the treewidth is more than $k$. Thus, we can find the minimum $k$ for which every vertex is deleted after some steps, and that $k$ is at most the treewidth of $G$.

\section{Conclusions and Related Work}
\label{l_related_work}
Cooperation is a central topic in algorithmic game theory. We considered computing the optimal coalition structure in WGGs. We showed that the problem is \np-hard, but restrictions on the input graph, such as being a tree or having bounded treewidth result in tractable algorithms for the problem. We showed the problem is hard for planar graphs, but provided a polynomial constant approximation algorithm for this class and other classes. 

WGGs are a well-known representation of cooperative games, and offer a simple and concise way of expressing synergies. One limitation of our approach is that some cooperative games cannot be expressed as a WGG. A general representation of a cooperative game is a table mapping any subset of the agents to the utility it can achieve (i.e. a table with size exponential in the number of the agents). Although the WGG representation is very concise for some games, and requires much less space than the exponential size table, some games cannot be expressed as WGGs. Further, even for games given in another representation language and that can be expressed as WGGs, there does not always exist a tractable algorithm for converting the game's representation to a WGG. To use our methods, one must have the input game given as a WGG. Since the WGG representation is a very prominent representation language for cooperative games, we believe our approach covers many important domains.  
%

Much work in algorithmic game theory has been dedicated to team formation, cooperative game representations and methods for finding optimal teams game theoretic solutions. 
Several papers describe representations of cooperative domains based on combinatorial structures~\cite{deng_papa_1994,IeoSho06,bachrach2010coalitional,ohta2006compact,bachrach2009power} 
and a survey in \cite{bilbao:2000:cooperative}. A detailed presentation of such languages 
is given in~\cite{coopgt-book,shoham2009multiagent}. 
Generation of the optimal coalition structure received much attention~\cite{shehory-kraus:1998a,sandholm1999coalition,rahwan2007anytime,rahwan2008improved} 
due to its applications, such as vehicle routing and multi-sensor networks. 
%
An early approach ~\cite{shehory-kraus:1998a} focused on overlapping coalitions and gave a loose approximation algorithm. Another early approach~\cite{sandholm1999coalition} has a worst case complexity of $O(n^n)$, whereas dynamic programming approaches~\cite{YunYeh1986} have a worst case guarantee of $O(3^n)$. Such algorithms were examined empirically in~\cite{larson2000anytime}.

Arguably, the state of the art method is presented in~\cite{rahwan2008improved}. It has a worst case runtime of $O(n^n)$ and offers no polynomial runtime guarantees, but in practice it is faster than the above methods. All these methods assume a black-box that computes the value of a coalition, while we rely on a specific representation. Another approach solves the coalition structure generation problem~\cite{bachrach2010coalitional}, but relies on a different representation. 
A fixed parameter tractable approach was proposed for typed-games~\cite{aziz2011complexity} (the running time is exponential in the number of agent ``types''). However, in graph games the number of agent types is unbounded, so this approach is untractable. 
%
%
In contrast to the above approaches, we provide polynomial algorithms and sufficient conditions that guarantee various \emph{approximation ratios} for WGGs~\cite{deng_papa_1994}.

This paper ignored the game theoretic problem of coalitional stability. While the structures we find do maximize the welfare, they do so in a potentially unstable manner. When agents are selfish, and only care about their own utility, the coalition structure may be broken when some agents decide to form a different coalition, improving their own utility at the expense of others. It would be interesting to examine questions relating to solution concepts such as the core, the nucleolus or the cost of stability\cite{oai:xtcat.oclc.org:OCLCNo/ocm19736643,schmeidler:1163,bachrach2009costAAMAS,bachrach2009costSAGT}.

Several directions remain open for future research. First, since solving the coalition structure generation problem is hard for general WGGs, are there alternative approximation algorithms? Obviously, one can use the algorithms for general games, however we believe a better complexity can be achieved for WGGs than for general games. Second, focusing on the game theoretic motivation for this work, it would be interesting to examine \emph{core-stable} coalition structures rather than just optimal ones. 
It is not known whether there exists a PTAS\footnote{Polynomial-Time Approximation Scheme.} for planar graphs or not. Though one might hope to obtain such a PTAS by combining Baker's 
approach with our algorithm for solving bounded treewitdh graphs, such a direct approach fails\footnote{See appendix for a complete discussion.}. 
Finally, it would be interesting to examine other classes of graphs where one can solve the coalition structure generation in polynomial time. Also, similar tractability results for coalition structure generation could be devised by using other representation languages.

\section{Acknowledgement}
The authors want to thank the anonymous reviewer who suggested applying Nash-Williams Theorem to partition the edges of graphs into forests. This simplified our previous analysis and improved the approximation guarantees of our algorithm. We also thank the same reviewer for providing an alternative algorithm for bounded degree graphs with the same approximation guarantee to get rid of the $\log \Delta$ factor in the running time. We presented our own algorithm as it has linear space complexity in the number of edges of the graph. Our algorithm can have up to a $\Delta$ factor less space complexity than the suggested approach.
 
\bibliographystyle{plain}
\bibliography{csgg-v4}


\section{Appendix I: Detailed Proofs}

{\bf Theorem} \ref{l_thm_opt_cs_hard_stsg}
OPT-CS-WGG is \np-complete, and assuming $\mathrm{P} \neq \mathrm{NP}$, there is no polynomial time $O(n^{1/2-\epsilon})$-approximation algorithm for it where $n$ is the number of vertices and $\epsilon$ is any positive constant. There is no polynomial time $O(n^{1-\epsilon})$-approximation algorithm for this problem unless $\mathrm{NP}=\mathrm{ZPP}$.

\begin{proof}
We reduce an Independent Set instance $G=\tup{V,E}$ to a OPT-CS-WGG instance. The reduced graph $G'=\tup{V',E'}$ contains all the vertices and edges of $G$ (so $V \subseteq V'$ and $E \subseteq E'$), and one additional vertex $s$. The weights of the original edges are identical, and are all $w(e) = -k$ where $k > |V|$. Vertex $s$ is connected to any vertex in $v \in V$ with an edge $e$ with weight $w(e) = 1$. Graph $G$ has an independent set of size $k$ iff there exists a partition for graph $G'$ with value $k$.

If $G$ has an independent set $S$ with size $k$, we can create a coalition structure $CS$ with value $k$ as follows.
We put (the additional) vertex $s$, and the $k$ vertices in set $S$ in one coalition. For every other vertex $v \notin S$, we make a separate coalition with only vertex $v$ in it, so the rest of the vertices are put in $n-k$ separate coalitions where $n$ is the number of vertices of graph $G$.
A coalition with one vertex has value zero. Thus the value of the whole coalition structure is equal to the value of the coalition that includes $s$. There is no negative edge between vertices of this coalition, because set $S$ is an independent set in graph $G$. There are $k$ edges with weight $+1$ in this coalition, so the value of the coalition structure is equal to $k$.

On the other hand, if the optimal coalition structure has value $k'$ in graph $G'$, we can find an independent set of size $k'$ in $G$.
We note that the absolute value of a negative edge is more than the sum of all positive edge weights because there are $n=|V|$ positive edges with weight $+1$. We have the option of putting each vertex in a separate coalition and achieving value zero. Thus, we never put a negative edge in a coalition in the optimal coalition structure, as this obtains a negative value. Thus, the value of the optimal coalition structure is the number of positive edges we get through the coalitions. However, all positive edges are between vertex $s$ and the other vertices, so the value of the coalitions is the number of vertices from graph $G$ that we put in the coalition of vertex $s$. We can safely put the rest of the vertices in separate coalitions as there are no positive edges between them, so there is no benefit in putting them in the same coalition. In the optimal solution the vertices we put in the same coalition as vertex $s$ have no negative edges between them, so they form an independent set in graph $G$.

From the analysis, we conclude that the above-mentioned is a parsimony reduction between the independent set problem and the OPT-CS-WGG problem. H{\aa}stad proved that there exists no polynomial time $O(n^{1/2-\epsilon})$-approximation algorithm for independent set problem assuming $\mathrm{P} \neq \mathrm{NP}$, he also proved that there is no polynomial time $O(n^{1-\epsilon})$-approximation algorithm for it unless $\mathrm{NP}=\mathrm{ZPP}$ \cite{Hastad96}. So the same hardness results work for the OPT-CG-WGG problem. For additional discussion of optimal coalition structures in graph games, see~\cite{aziz2011complexity}.
\end{proof}

\noindent {\bf Theorem} \ref{theorem_planar_hardness}
OPT-CS-WGG in planar graphs is \np-complete.

\begin{proof}
We use a reduction from the Independent Set problem restricted to planar graphs, which is also \np-complete \cite{Garey76}.
Given a planar graph $G=\tup{V,E}$, we transform the corresponding Independent Set instance it to the following planar OPT-CS-WGG instance $\tup{V',E',w}$.
Each node $i\in V$ of degree $d$ is replaced with gadget $G_i$ with $2d+1$ nodes and $4d$ edges.
This gadget can be thought of as a $d$-sided polygon whose sides correspond to edges in ${\cal N}(i)\equiv\{(i,j)\in E\}$
(figure~\ref{fig:planar:i}(a) shows an example for $d=4$). The side corresponding to $e$ has nodes $i^-_e,i_e,i^+_e\in V'$
arranged in the clockwise order. Note, if $e,e'\in E$ are two consecutive edges in ${\cal N}(i)$ taken in the clockwise
order then $i^+_e=i^-_{e'}$. Gadget $G_i$ also has the center node $i\in V'$ connected to all other $2d$ nodes of $G_i$.
The edge weights are as follows: $w(i^-_e,i_e)=w(i_e,i^+_e)=-dC$ where $C$ is a sufficiently large constant (namely, $C>|V|$),
and the weights of edges from $i$ to other nodes in $G_i$ are $+C$, except for one edge $(i,i^+_e)$ which has weight $C+1$.
Here edge $e\in{\cal N}(i)$ can be chosen arbitrarily.

Each edge $e=(i,j)\in E$ is transformed to gadget $G_e$ as shown in figure~\ref{fig:planar:e}.
Note, $G_e$ involves nodes $i^-_e,i^+_e,j^-_e,j^+_e$ and some internal nodes
which are connected only to nodes in $G_e$.

Given a coalition structure $CS$ in $\tup{V',E'}$, we denote $v_i(CS)$ and $v_e(CS)$ to be the contribution to the total cost
of subgraphs $G_i$ and $G_e$ respectively; thus, $v(CS)=\sum_{i\in V}v_i(CS)+\sum_{e\in E}v_e(CS)$.
We also denote $v^\ast_i=\max v_i(CS)$ and $v^\ast_e=\max v_e(CS)$. We say that coalition structure $CS$
is {\em feasible for $G_i$} if $v_i(CS) > v^\ast_i-C$, and it is {\em feasible for $G_e$} if $v_e(CS)>v^\ast_e-C$.
We denote $\sim$ to be the equivalence relation on $V'$ induced by $CS$.

\begin{figure}[t]
\centering
\begin{tabular}{@{\hspace*{0pt}}c@{\hspace*{1pt}}c@{\hspace*{-4pt}}c@{\hspace*{0pt}}c@{\hspace*{-17pt}}c}
%

%
\raisebox{20pt}{\includegraphics[scale=0.1]{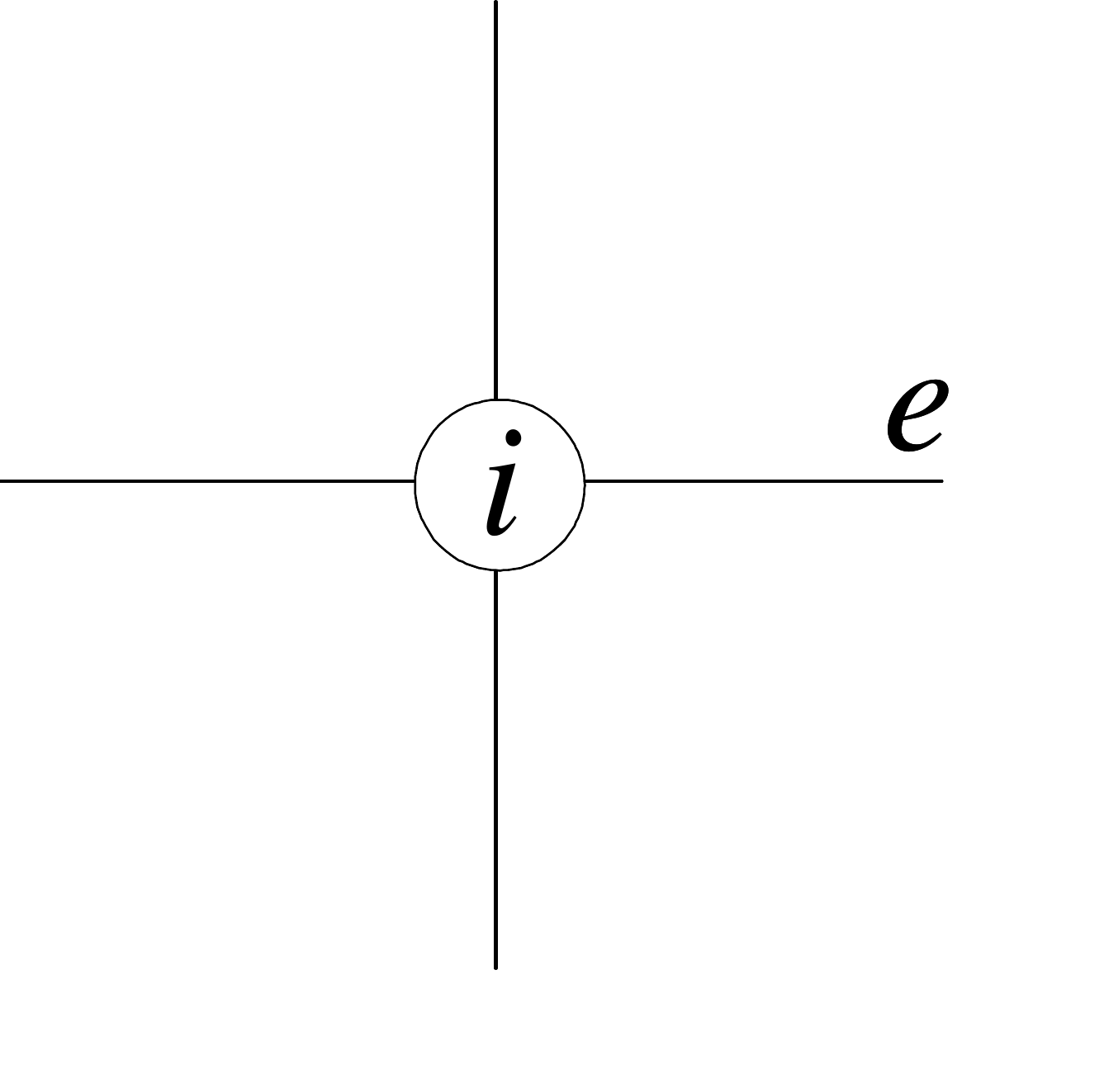}} &
\raisebox{30pt}{\includegraphics[scale=0.2]{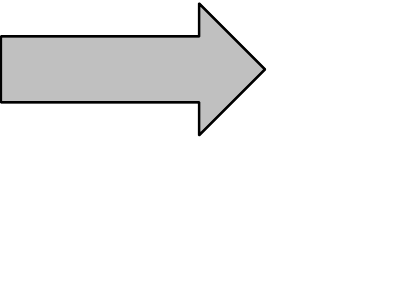}} &
\includegraphics[scale=0.4]{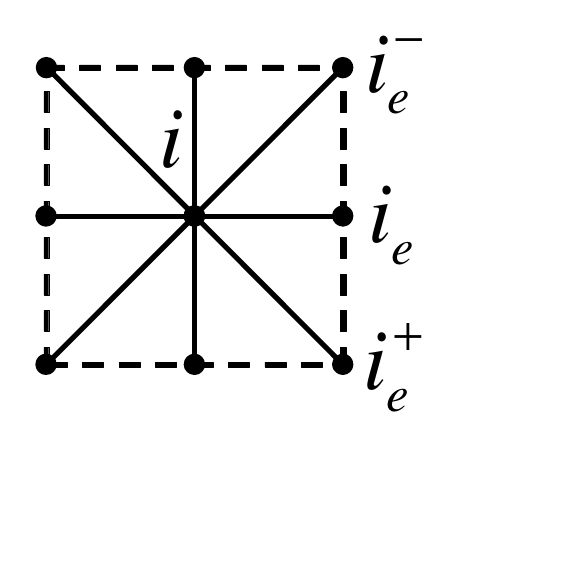} &
\includegraphics[scale=0.4]{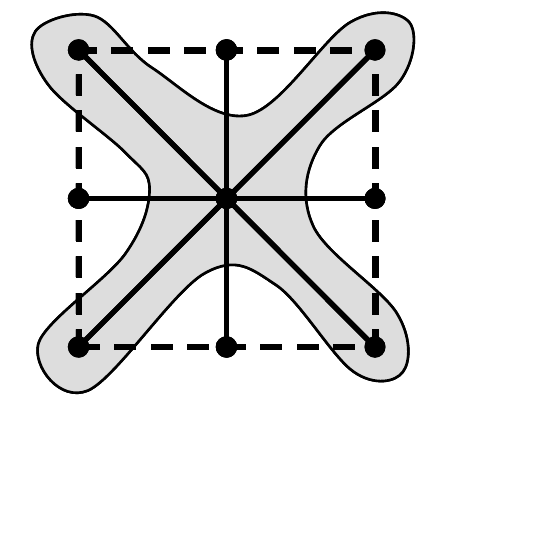} &
\includegraphics[scale=0.4]{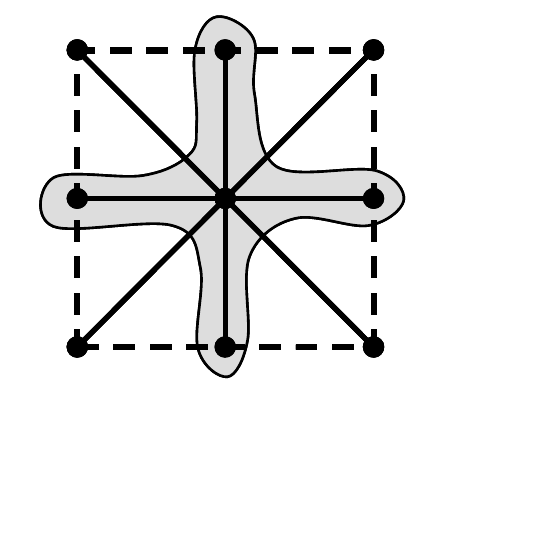} \vspace{-25pt} \\
 & (a) & & (b)~~~~~ & (c)~~~~~~
\end{tabular}
\caption{(a) node $i\in V$ of degree $d=4$ transformed to gadget $G_i$.
(b),(c) are feasible partitionings of $G_i$. (b) encodes the event that $i$ is in the independent set, while (c) encodes the opposite event.
}
\label{fig:planar:i}
\end{figure}

Consider node $i\in V$ of degree $d$. It can be seen that $v^\ast_i=dC+1$, and $CS$ is feasible for $G_i$
iff exactly one of the following holds: (1) $i^+_e\sim i$ and $i_e\nsim i$ for all $e\in{\cal N}(i)$
(in which case $v_i(CS)=v^\ast_i$),
or (2) $i^+_e\nsim i$ and $i_e\sim i$ for all $e\in{\cal N}(i)$ (in which case $v_i(CS)=v^\ast_i-1$).
These cases are illustrated in figure~\ref{fig:planar:i}(b,c).
Case (1) will encode the event that $i$ is in the independent set,
and case (2) will encode the opposite event.

Let $S^\ast$ be a maximum independent set in $\tup{V,E}$.
We define coalition structure $CS^\ast$ as follows:
(i) for each $i\in V$ select a feasible partitioning of $G_i$ according to $S$;
nodes of $G_i$ not connected to the center node $i$ are assigned to singleton
partitions. Partitions in graphs $G_i,G_j$ for $i\ne j$ do not overlap.
(ii) For each $e=(i,j)\in E$ select partitioning of $G_e$
that does not affect the equivalence relations between nodes in $\{i^-_e,i^+_e,j^-_e,j^+_e\}$
and has the maximum possible value of $v_e(CS^\ast)$ among such partitionings.

\begin{figure}[t]
\centering
\begin{tabular}{@{\hspace*{5pt}}c@{\hspace*{5pt}}c@{\hspace*{5pt}}c}
\includegraphics[scale=0.15]{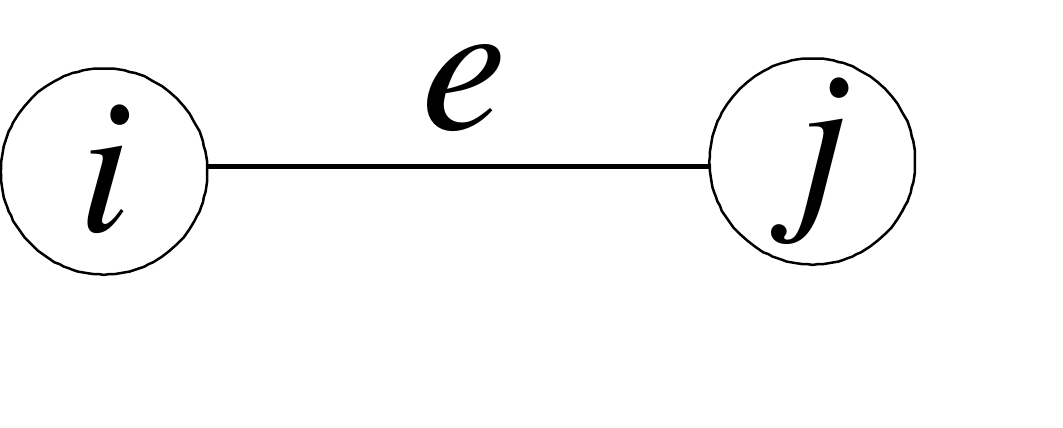}&
\raisebox{5pt}{\includegraphics[scale=0.15]{arrow.pdf}} &
\raisebox{-35pt}{\includegraphics[scale=0.45]{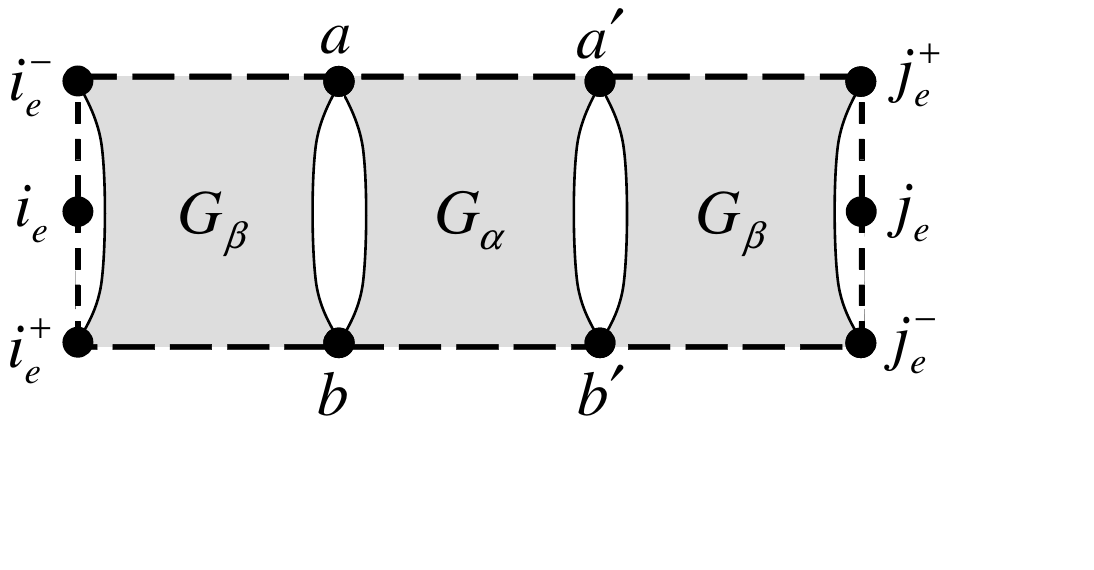}} \vspace{-20pt}

%
\end{tabular}
\caption{Each edge $(i,j)\in E$ is transformed to gadget $G_e$ consisting of three subgraphs $G_\beta$, $G_\alpha$, $G_\beta$
of figure~\ref{fig:planar:g}.
}
\label{fig:planar:e}
\end{figure}

\begin{figure}[t]
\centering
\begin{tabular}{c@{\hspace*{-10pt}}c@{\hspace*{-20pt}}c@{\hspace*{-20pt}}c}
\includegraphics[scale=0.50]{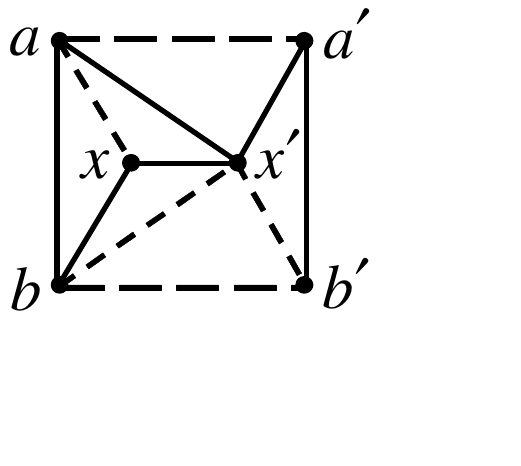} &
\includegraphics[scale=0.50]{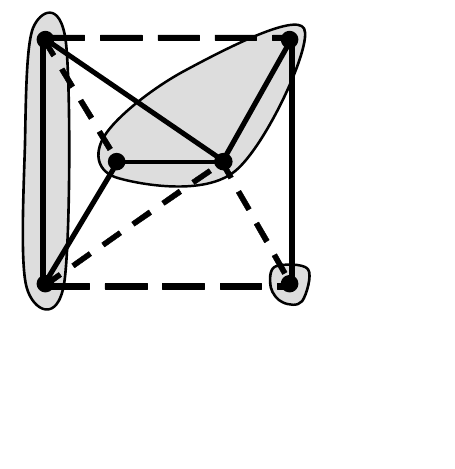} &
\includegraphics[scale=0.50]{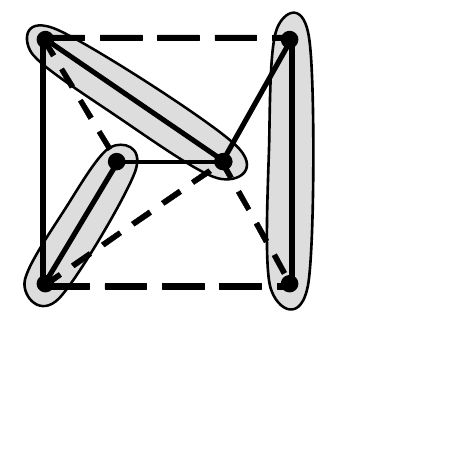} &
\includegraphics[scale=0.50]{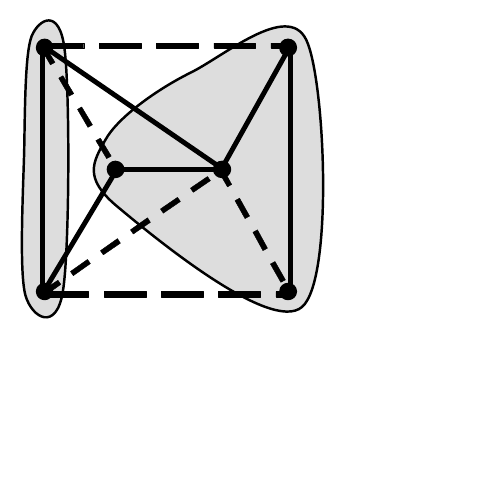} \vspace{-25pt} \\

(a)~~~~~~~~~ & (b)~~~~~~~~~ & (c)~~~~~~~~~ & (d)~~~~~~~~~
\end{tabular}
\caption{(a) Topology of graphs $G_\alpha$ and $G_\beta$. The weights in $G_\alpha$ are as follows:
for solid edges
$w(a,b)=w(b,x)=w(x,x')=w(x',a)=w(x',a')=w(a',b')=+2C$,
for internal dashed edges
$w(a,x)=w(b,x')=w(x',b')=-2C$,
and for horizontal side edges $w(a,a')=w(b,b')=-8C$.
The weights in $G_\beta$ are the same except for two edges: $w(a,b)=w(a',b')=+C$.
(b),(c),(d) are optimal partitionings of $G_\alpha$,
and (b),(c) are optimal partitionings of $G_\beta$.
}
\label{fig:planar:g}
\end{figure}

Consider edge $e=(i,j)\in E$. We prove two facts:
\begin{itemize}\addtolength{\labelwidth}{-12pt}
\item[(a)] $v_e(CS^\ast)=v^\ast_e$,
and thus 
$
v(CS^\ast)
=\sum_{i\in S^\ast} v^\ast_i + \sum_{i\in V-S^\ast} (v^\ast_i-1) + \sum_{e\in E} v^\ast_e 
=\sum_{i\in V} v^\ast_i + \sum_{e\in E} v^\ast_e + |S^\ast| - |V|.
$
\item[(b)] If $CS$ is feasible for $G_i$, $G_j$ and $G_e$ then $i$ and $j$ cannot be both in the independent set,
i.e.\ either $i\sim i_e$ or $j\sim j_e$.
\end{itemize}
This will imply that $CS^\ast$ is an optimal coalition structure, and so solving constructed OPT-CS-WGG instance
will recover $|S^\ast|$. Indeed,
if $v(CS)\ge v(CS^\ast)$ for some $CS$ then $v(CS)>\sum_{i\in V} v^\ast_i + \sum_{e\in E} v^\ast_e - C$.
Thus, $v_i(CS) > v^\ast_i-C$ for all $i\in V$ and $v_e(CS) > v^\ast_e-C$ for all $e\in V$,
and so $CS$ is feasible for all gadgets $G_i$ and $G_e$.
This means that $CS$ defines some set $S\subseteq V$, and by property (b) this set is independent, implying
$v(CS)
\le \sum_{i\in S} v^\ast_i + \sum_{i\in V-S} (v^\ast_i-1) + \sum_{e\in E} v^\ast_e 
= \sum_{i\in V} v^\ast_i + \sum_{e\in E} v^\ast_e + |S| - |V| \le v(CS^\ast)$.

To prove (a) and (b), we first analyze optimal partitionings of graphs
$G_\alpha$ and $G_\beta$ specified in figure \ref{fig:planar:g}.

\underline{Graph $G_\alpha$}~~ Let $v_\alpha(CS)$ be the contribution of $G_\alpha$ to the total cost.
If the restriction of $CS$ to $G_\alpha$ is as shown in figure \ref{fig:planar:g}(b,c,d) then $v_\alpha(CS)=6C$.
This is the maximum possible cost, i.e.\ $v^\ast_\alpha\equiv\max v_\alpha(CS)=6C$. Indeed, the subgraph induced by nodes $\{a,b,x,x'\}$ cannot
contribute cost larger than $4C$, and the subgraph induced by nodes $\{a',b',x'\}$ cannot
contribute cost larger than $2C$.
Clearly, we must have $a\nsim a'$ and $b\nsim b'$, otherwise $v_\alpha(CS)\le v^\ast_\alpha-2C$.
Assuming that these conditions hold, $4$ cases are possible:
\begin{itemize}
\item[] (i) $a\sim b$, $a'\nsim b'$;
\item[] (ii) $a\nsim b$, $a'\sim b'$;
\item[] (iii) $a\sim b$, $a'\sim b'$;
\item[] (iv) $a\nsim b$, $a'\nsim b'$.
\end{itemize}
In the first $3$ cases there exist a partitioning $CS$ of $G_\alpha$
with $v_\alpha(CS)=v^\ast_\alpha$ (see figure \ref{fig:planar:g}(b,c,d)).
In case (iv), however, we have $v_\alpha(CS)\le v^\ast_\alpha-2C$.
Indeed, edges $(a,b)$ and $(a',b')$ are not contributed, and subgraphs induced by $\{a,b,x,x'\}$ and
$\{a',b',c'\}$ contribute not more than $4C$ and $2C$ respectively. Furthermore, if the first subgraph
contributes more than $2C$ then we must have $a\sim x'$ and thus $a'\nsim x'$, so the second subgraph does not contribute.

\underline{Graph $G_\beta$}~~ Again we must $a\nsim a'$ and $b\nsim b'$, otherwise $v_\beta(CS)\le v^\ast_\beta-C$.
Assume that these conditions hold. Then in cases (i), (ii) there exist a partitioning $CS$ of $G_\beta$
with $v_\beta(CS)=v^\ast_\beta$ (as shown in figure \ref{fig:planar:g}(b,c)),
and in cases (iii), (iv) we have $v_\alpha(CS)\le v^\ast_\alpha-C$.
This can derived by observing that $G_\beta$ is obtained from $G_\alpha$ by subtracting weight $C$ from edges $(a,a')$ and $(b,b')$,
and applying the analysis for $G_\alpha$ above.

We are now ready to prove properties (a) and (b). We use the notation from figure~\ref{fig:planar:e}.
It can be checked $v^\ast_e=v^\ast_\beta+v^\ast_\alpha+v^\ast_\beta$, and we have the following:
$CS$ is feasible for $G_e$ $\quad\Leftrightarrow\quad$ $v_e(CS)=v^\ast_e$ $\quad\Leftrightarrow\quad$
$v^\ast_\alpha(CS)=v_\alpha$, $v^\ast_\beta(CS)=v_\beta$ for both copies of graph $G_\beta$ involved in $G_e$.

\vspace{3pt}
\noindent {\bf (a)} Let $\sim$ be the equivalence relation induced by $CS^\ast$.
Three cases are possible:
\begin{itemize}
\item $i\notin S^\ast$, $j\notin S^\ast$. Then $i^-_e\!\nsim i^+_e$, $j^-_e\!\nsim j^+_e$,
and so we can choose a partitioning $CS$ of $G_e$ with $v_e(CS)=v^\ast_e$ by compositing optimal partitionings (c), (d), (b) from figure~\ref{fig:planar:g}
for the three subgraphs $G_\beta$, $G_\alpha$, $G_\beta$ respectively. Clearly, these partitionings are consistent with each other; we have $a\sim b$ and $a'\sim b'$.
\item $i\in S^\ast$, $j\notin S^\ast$. Then $i^-_e\!\sim i^+_e$, $j^-_e\!\nsim j^+_e$,
and so we can choose a partitioning $CS$ of $G_e$ with $v_e(CS)=v^\ast_e$ by compositing optimal partitionings (b), (c), (b) from figure~\ref{fig:planar:g}
for the three subgraphs $G_\beta$, $G_\alpha$, $G_\beta$ respectively. Clearly, these partitionings are consistent with each other; we have $a\nsim b$ and $a'\sim b'$.
\item $i\notin S^\ast$, $j\in S^\ast$. Then $i^-_e\!\nsim i^+_e$, $j^-_e\!\sim j^+_e$,
and so we can choose a partitioning $CS$ of $G_e$ with $v_e(CS)=v^\ast_e$ by compositing optimal partitionings (c), (b), (c) from figure~\ref{fig:planar:g}
for the three subgraphs $G_\beta$, $G_\alpha$, $G_\beta$ respectively. Clearly, these partitionings are consistent with each other; we have $a\sim b$ and $a'\nsim b'$.
\end{itemize}


\vspace{3pt}
\noindent {\bf (b)}
Suppose that $CS$ is feasible for $G_i,G_j,G_e$ and $i,j$ are both in the independent set implying $i^-_e\sim i^+_e$
and $j^-_e\sim j^+_e$. For graphs $G_\beta$ we can only have cases (i) and (ii), therefore $a\nsim b$ and $a'\nsim b'$.
Thus, we have case (iv) for graph $G_\alpha$ - a contradiction.
\end{proof}

\section{Appendix II: Why Baker's Approach~\cite{Baker94} Does Not Yield a PTAS for Planar Graphs}
Baker's approach \cite{Baker94} and other decomposition theorems \cite{Demaine05} have been used to partition vertices or edges of a planar graph and more generally a minor free graph into an arbitrary number of sets such that removing any of these sets results a bounded treewidth graph. Many problems are tractable for bounded treewidth graphs, so this technique is quite powerful for domains where each set in the decomposition is small enough so that deleting it does not cause extreme changes in the optimal solution. One might hope that the same approach would allow obtaining a PTAS for the optimal coalition structure in planar graphs. We show why a direct solution based on Baker's approach can not be used to get a PTAS, and discuss why other decomposition theorems are not helpful in this setting as well. 

We first briefly explain Baker's approach of finding disjoint subsets of edges of the graph such that the deletion of any of these sets of edges results a bounded treewidth graph. 
Take an arbitrary vertex $v$, and consider the BFS tree rooted at $v$. Let $L_i$ be the set of vertices at distance $i$ from $v$, let $l$ be the number of these sets, i.e. $l$ is the maximum distance from $v$. These sets are called ``layers''.  For any integer number $k$, and $i<k$, if we remove the edges between sets $L_{i+jk}$ and $L_{1+i+jk}$ for all $0 \leq j \leq l/k$, the resulting graph would have treewidth $O(k)$ (see~\cite{Demaine04} for a proof). Therefore removing edges between layers $L_{i+jk}$ and $L_{1+i+jk}$ for all $0 \leq j \leq l/k$ gives us a bounded treewidth graph for which we have a polynomial time exact algorithm. 

Let $E_i$ be the set of edges between layers $L_{i+jk}$ and $L_{1+i+jk}$ for all $0 \leq j \leq l/k$. Baker's approach provides $k$ disjoint sets $E_0, E_1, \ldots, E_{k-1}$ such that deleting any $E_i$ gives us an $O(k)$ bounded treewidth graph. If the sets generated by Baker's algorithm were a partition, we would get a PTAS for planar graphs. Unfortunately, these $k$ sets do not form a partition of edges of the graph. They are disjoint, but their union is \emph{not} all the edges, since there may be edges inside each layer which are not in any of these $k$ sets.  We now discuss how such an approach would work if $\{E_i\}_{i=0}^{k-1}$ were a partition and explain where it fails when $\{E_i\}_{i=0}^{k-1}$ is not a partition.

Consider the case where the above $\{E_i\}_{i=0}^{k-1}$ is a partition. For any $i$, we can modify the optimal solution to avoid all edges in set $E_i$, while decreasing the solution's quality. Assume that the optimum solution partitions the vertices of the graph into $t$ sets $S_1, S_2, \ldots, S_t$. Let $E_{opt}$ be the set of edges covered in this optimum solution. When removing the edges in $E_i$ we obtain a graph with some connected components. In each connected component, we can use the optimum solution's partitioning to partition the vertices. We use the projection of the optimum solution's partition in each connected component. In this new partitioning, we cover all edges that the optimum solution covers except the ones in $E_i$. Note that this new partition is a feasible solution in the graph $G \setminus E_i$ (the graph $G$ after removing the edges in $E_i$). Since graph $G \setminus E_i$ has bounded treewidth, we can find its optimum solution in polynomial time. This solution has value at least $\sum_{e \in E_{opt}} w(e) - \sum_{e \in E_{opt} \cap E_i} w(e)$ (the value of the solution derived by projecting the optimum solution in each connected component of $G \setminus E_i$). We can find the optimum solution in $G \setminus E_i$ for each $0 \leq i < k$. The average of these $k$ solutions is at least:
 $
 \frac{\sum_{i=0}^{k-1} \Big{[} \sum_{e \in E_{opt}} w(e) - \sum_{e \in E_{opt} \cap E_i} w(e) \Big{]}}{k} = 
  \frac{k \sum_{e \in E_{opt}} w(e) - \sum_{i=0}^{k-1} \sum_{e \in E_{opt} \cap E_i} w(e)}{k}
 $  
 
Unfortunately, the sets $\{E_i\}_{i=0}^{k-1}$ are disjoint but do not form a partition. An approach based on Baker's algorithm works when $\sum_{i=0}^{k-1} \sum_{e \in E_{opt} \cap E_i} w(e)$ is at most $\sum_{e \in E_{opt}} w(e)$.
However, this may not hold since there are negative edges in the graph, and possibly in the optimum solution as well. It might be the case that the positive edges in the optimum solution are shared with sets $E_i$'s, and the negative edges of the optimum solution are not in any set $E_i$. In this case, we can not obtain any good upper bound on $\sum_{i=0}^{k-1} \sum_{e \in E_{opt} \cap E_i} w(e)$, and it may even be larger than $k \sum_{e \in E_{opt}} w(e)$, so we may gain nothing through this approach. 
 
The key problem in the above approach is that the sets of edges in Baker's approach do not form a partition of edges of the graph: they are disjoint sets, but do not span all the edges. In particular, the edges between two consecutive layers are covered by Baker's approach, and the edges inside each layer are not present in any set. Thus, if all edges between two layers have weight $+\infty$ and the edges inside each layer have weight $-\infty$, we can not find any reasonable upper bound on $\sum_{i=0}^{k-1} \sum_{e \in E_{opt} \cap E_i} w(e)$ in terms of the optimum solution.  

Other decomposition theorems, such as~\cite{Demaine05}, are also not very helpful for similar reasons. For example, Theorem 3.1 in~\cite{Demaine05} states that we can partition the vertices of a minor free graph into $k$ sets such that the deletion of any of these sets results a bounded treewidth graph. One might hope there exists one of these $k$ sets whose deletion does not effect the optimal solution by a factor of more than $1-1/k$. A counterexample is when all positive edges in the optimal solution are between these $k$ sets, and the negative edges of the optimal solution are inside these $k$ sets. Again, the absolute value of the weights of all edges may be much larger than the value of the optimal solution. In this case, removing these $k$ sets affects the positive edges more than the negative edges, which can have a huge effect the value of the solution. In particular, if we remove each of these $k$ sets once, every (positive) edge between two sets is ``removed'' twice (for each of its endpoints), and every (negative) edge inside a set is ``removed'' only once.

\end{document}